\documentclass{article}

\usepackage{amsmath,amsfonts,amssymb,amsthm}
\usepackage[round]{natbib}
\usepackage{mathtools} 
\usepackage{color}
\usepackage[margin=1.5in]{geometry}
\usepackage{comment}
\usepackage{multirow}
\usepackage{verbatim}
\usepackage{bbm}
\usepackage{booktabs}
\usepackage{float}
\usepackage{graphicx}
\usepackage{paralist} % for compactitem, compactdesc, compactenum
\usepackage{etoolbox}
\usepackage{authblk}

\newcommand{\p}{\mathbb{P}}
\newtheorem{theorem}{Theorem}

\newcommand{\Bern}{\text{Bern}}

\title{Recalibration of Predictive Models as \\Approximate Probabilistic Updates}

\author{
  Evan T. R. Rosenman \\ \vspace{-3mm}
  \emph{Data Science Initiative}\\\emph{Harvard University}
    \and
  Santiago Olivella \\ \vspace{-3mm}
  \emph{Department of Political Science} \\\emph{University of North Carolina, Chapel Hill}
}

\begin{document}

\maketitle

\begin{abstract}
The output of predictive models is routinely recalibrated by reconciling low-level predictions with known derived quantities defined at higher levels of aggregation. For example, models predicting turnout probabilities at the individual level in U.S. elections can be adjusted so that their aggregation matches the observed vote totals in each state, thus producing better calibrated predictions. In this research note, we provide theoretical grounding for one of the most commonly used recalibration strategies, known colloquially as the ``logit shift''. Typically cast as a heuristic optimization problem (whereby an adjustment is found such that it minimizes the difference between aggregated predictions and the target totals), we show that the ``logit shift'' in fact offers a fast and accurate approximation to a principled, but often computationally impractical adjustment strategy: computing the posterior prediction probabilities, conditional on the target totals. After deriving analytical bounds on the quality of the approximation, we illustrate the accuracy of the approach using Monte Carlo simulations. The simulations also confirm analytical results regarding scenarios in which users of the simple logit shift can expect it to perform best --- namely, when the aggregated targets are comprised of many individual predictions, and when the distribution of true probabilities is symmetric and tight around 0.5. \\
%\keywords{Recalibration, Poisson-Binomial distribution, Logit shift, Election prediction}
\end{abstract}

\section{Problem Description}

A common problem in predictive modeling is that of calibrating probabilities to observed totals. For example, an analyst may obtain individual-level scores $p_i \in (0, 1), i = 1, \dots, N,$ to estimate the probability that each of the $N$ registered voters in a particular voting precinct will support the Democratic candidate in an upcoming election. 
After the election occurs, the analyst can observe the total number of Democratic votes, $D$, cast among the subset $\mathcal{V} \subset \{1, \dots, N\}$ of registered voters who cast a ballot. But she cannot observe individual-level outcomes due to the secret ballot. In the absence of perfect prediction, the analyst will find that $\sum_{i \in \mathcal{V}} p_i \neq D$. She must then decide how to compute recalibrated scores, $\tilde p_i$, to better reflect the realized electoral outcome. 

This practical problem has direct implications for public opinion research. For example, \cite{ghitza2020voter} recalibrate their MRP estimates of voter support levels after an election to match county-level totals, while \cite{schwenzfeierPolMeth} proposes using the magnitude of the calibration to estimate non-response bias in public opinion polls. The problem is also of great importance in campaign work. Campaigns frequently seek to target voters who are most likely to have supported their party in the prior presidential election. Estimates of prior party support may also serve as predictor variables in models estimating support in successive elections. Recalibrating the scores to match observed outcomes is thus a crucial step to improve the scores' accuracy and bolster future electioneering. 
 
A common heuristic solution to the recalibration problem is the use of the ``uniform swing'' \citep{butler1951appendix} on the logit scale. This approach is simple: first, one defines the function 
\[ h(\alpha) = \sum_{i \in \mathcal{V}} \frac{1}{1 + \frac{1-p_i}{p_i} \alpha} \,,\] 
and, having observed a total $D$, one finds the $\alpha$ that satisfies the equation 
\begin{equation}\label{eq:basicSwing}
    h(\alpha)  = D\,.
\end{equation}
The function $h(\cdot)$ is monotonic in $\alpha$, so Equation \ref{eq:basicSwing} can be solved in logarithmic time using binary search. The updated scores are then computed as 
\[ \tilde p_i =  \frac{1}{1 + \frac{1-p_i}{p_i} \alpha} \,. \]
Solving Eq.~\ref{eq:basicSwing} is equivalent to finding the set of probabilities $\tilde p_i$ which sum to $D$ and minimize the Kullback–Leibler divergence \citep{kullback1951information} with the distribution induced by the original scores $p_i$. 
%$\alpha$ that minimizes the absolute difference between $\tilde{p}_i$ and $D$. 
Moreover, if the $p_i$ are defined based on a logistic regression, then this update is equivalent to shifting the intercept in the model by $\log(\alpha)$. For more details on these characterizations, see the Appendix, Section \ref{sec:char}.

Examples of this simple recalibration strategy are given by \cite{ghitza2013deep}, \cite{hanrettyetal2016}, and \cite{ghitza2020voter}. This procedure is also familiarly referred to by campaign workers as the ``logit shift".\footnote{The term ``logit swing" is also commonly used.}
%Anecdotally, it is also routinely used by campaign workers, where the procedure is commonly known as a ``logit shift.'' 

In this research note, we provide analytical justification for the logit shift. First, we introduce an alternative procedure for score updating, which simply computes the updated scores as posterior probabilities, conditional on the target totals. In this procedure, we assume the original scores $p_i$ capture a kind of prior Democratic support probability, while the updated scores $\tilde p_i$ reflect the conditional Democratic voting probability given observed outcomes. Next, we show that this second, more principled approach is well approximated by the logit shift in large samples. We demonstrate this result analytically and illustrate it in a small simulation study. Finally, we discuss potential extensions to cases where a uniform swing is insufficient to capture observed electoral dynamics. 

%Suppose we have a population of units $i = 1, \dots, N$, and an associated binary outcome $W_i \in \{0, 1\}$. Prior to the realization of the $W_i$, we compute a set of estimated probabilities for every member of the population $p_i, i = 1, \dots, N$. The $p_i$ are derived by modeling the $W_i$ as Bernoulli random variables, and the $p_i$ are meant to approximate $P(W_i = 1)$ for each individual (hence capturing some intrinsic probability for each individual). After the realization of the $W_i$, we do not observe the $W_i$ directly,but we do observe $D = \sum_i W_i$. We would like to find a set of scores $p_i^{\star}$ that are close analogues to $p_i$ but add up to the observed sum $D$.

%A standard application of this problem can be found in political science. Suppose the $p_i$ are an estimate of Democratic voting probabilities fit prior to an election. Each score estimates the individual's probability of voting Democratic. When the election occurs, we observe that the Democratic vote total $D$ among the $N$ voters is lower than the expected vote total $\sum_i p_i$. We want to ``update" the scores to better reflect the observed vote totals, but want to retain most of the structure of our initial scores. 

\section{Recalibration as a posterior update}

To motivate the posterior update approach, we introduce some additional notation. We define each voter's choice as a binary variable $W_i \in \{0, 1\}$, where $W_i = 1$ signifies a Democratic vote and $W_i = 0$ signifies a Republican vote (we suppose a two-candidate election for simplicity). The $W_i$ are modeled as independent Bernoulli random variables, where $W_i \sim \Bern(p_i)$. In other words, the $p_i = P(W_i = 1)$ can be thought of as the prior, unconditional probability of casting a Democratic vote. 
In this model, it is straightforward to approach score recalibration by defining a new set of updated scores, $\{p_i^{\star}\}$, using the following conditional probability (which automatically sum to $D$ over voters $i$):
% \[ p_i^{\star} = \p \left( W_i = 1 \,\middle\vert\, \sum_{j \in \mathcal{V}} W_j = D \right). \] 

% Note that such scores will automatically sum to $D$, because 
% \begin{align*}
% \sum_{i \in \mathcal{V}}  p_i^{\star} &= \sum_{j \in \mathcal{V}} \p \left( W_i = 1 \mid \sum_{j \in \mathcal{V}} W_j = D \right) \\
% &= \sum_{i \in \mathcal{V}} \E \left( W_i  \mid \sum_{j \in \mathcal{V}} W_j = D \right) =  \E \left( \sum_{i \in \mathcal{V}} W_i  \mid \sum_{j \in \mathcal{V}} W_i = D \right)  = D.
% \end{align*} 
\begin{align}
\begin{split}
p_i^{\star} &= \p \left( W_i = 1 \,\middle\vert\, \sum_{j \in \mathcal{V}} W_j = D \right)\\
&=\frac{\p \left(W_i = 1, \sum_{j \in \mathcal{V}} W_j = D\right)}{\p\left(\sum_{j \in \mathcal{V}} W_j = D\right)}\\ 
&= p_i \times \xi_i
\end{split}
\label{eq:posterior}
\end{align}
where $\xi_i=\frac{\p\left( \sum_{j \neq i} W_j  = D - 1\right)}{\p\left(\sum_j W_{j \in \mathcal{V}} = D\right)}$
is a ratio of two Poisson-Binomial probabilities --- that is, probabilities over the sum of independent \emph{but not identically distributed} Bernoulli random variables \citep{chen1997statistical}. Explicit computation of the $p_i^{\star}$ is quite challenging, as efficient computation of Poisson-Binomial probabilities is extremely computationally demanding at even moderate sample sizes, despite substantial recent advances in the literature \citep{olivella2017poisbinom, junge2020package}. To compute the $p_i^{\star}$, we would need to compute one unique Poisson-Binomial probability per unit in the population. Hence, if the number of actual voters $|\mathcal{V}|$ were even modestly large, it would be computationally infeasible to obtain these exact posterior probabilities. 

\section{Logit shift approximates the correct posterior}

\subsection{Preliminaries}

In this section, we show analytically why the logit shift is a good approximation to the general posterior update in Eq.~\ref{eq:posterior}. 
To do so, we begin by defining two terms, the ratio $\phi_i = \frac{\p\left(\sum_{j \neq i} W_j = D\right)}{\p\left(\sum_{j \neq i} W_j = D-1\right)}$, and the function $f(x, s)  = \frac{1}{1 + \frac{1-x}{x}(s)} \,. $

Simple substitution, along with a useful recursive property of the Poisson-Binomial distribution,\footnote{\label{fn:recursion}Namely, 
\begin{equation*}
\p\left(\sum_j W_j = D \right) = p_i\times\p\left(\sum_{j \neq i} W_j = D-1\right) + (1-p_i)\times\p\left(\sum_{j \neq i} W_j = D\right) \,. 
\end{equation*} } makes it clear that
\begin{align}
\begin{split}
\sum_if(p_i, \phi_i) &= \sum_i\frac{1}{1 + \frac{1-p_i}{p_i} \phi_i}\\
& =  \sum_i\frac{1}{1 + \frac{1-p_i}{p_i}  \frac{\p\left(\sum_{j \neq i} W_j = D\right)}{\p\left(\sum_{j \neq i} W_j = D-1\right)} } \\
%&= \frac{1}{ \frac{p_i\p\left(\sum_{j \neq i} W_j = D-1\right) + (1-p_i)\p\left(\sum_{j \neq i} W_j = D\right)}{p_i\p\left(\sum_{j \neq i} W_j = D-1\right)} } \\
&= \sum_i\frac{p_i\times\p\left(\sum_{j \neq i} W_j = D-1\right) }{ p_i\times\p\left(\sum_{j \neq i} W_j = D-1\right) + (1-p_i)\times\p\left(\sum_{j \neq i} W_j = D\right) } \\
&=  \sum_i\frac{\p \left(W_i = 1, \sum_i W_i = D\right)}{\p\left(\sum_i W_i = D\right)} \\
& =\sum_i p_i^{\star}\\
&=D
\end{split}
\label{eq:indswing}
\end{align}

In words, Eq.~\ref{eq:indswing} shows that $\phi_i$ is precisely the ``shift" that turns each $p_i$ into the desired, recalibrated posterior probability $p_i^{\star}$. The logit shift, however, uses a constant $\alpha$ to approximate the vector of recalibrating shifts $\{\phi_i\}_{i \in \mathcal{V}}$.
 What remains, therefore, is to show that the value of $\alpha$ that solves Eq.~\ref{eq:basicSwing} is a very good approximation of $\phi_i$ for all values of $i$. 

% Now, we define a generalization of $f(\cdot, \cdot)$ with vector inputs, 
% \[ g\left(\boldsymbol x, \boldsymbol s\right) = \sum_i f(x_i, s_i)  = \sum_i \frac{1}{1 + \frac{1-x_i}{x_i} s_i} \,.\] 
% Denote the vectors 
% \[ \boldsymbol p = \{p_i\}_{i \in \mathcal{V}},  \hspace{5mm} \boldsymbol{\phi} = \{\phi_i\}_{i \in \mathcal{V}}, \hspace{5mm} \text{ and } \hspace{5mm} 
% \boldsymbol \alpha = \{\alpha\}_{i \in \mathcal{V}}\,.
% \] 
% We observe that $g(\boldsymbol p, \boldsymbol \alpha) = h(\alpha)$. Moreover, we observe 
% \begin{align*}
% g(\boldsymbol p, \boldsymbol \phi) &= \sum_i f(p_i, \phi_i) = \sum_i p_i^{\star} = D\,. 
% \end{align*}

To do so, we establish a couple of facts: that the value of $\alpha$ is bounded by the range of $\{\phi_i\}_{i \in \mathcal{V}}$, and that each $\phi_i$ in turn has well-defined bounds:

\begin{theorem}\label{thm:boundsOnAlpha}
The value of $\alpha$ which solves Equation \ref{eq:basicSwing} satisfies: 
\[ \min_i \frac{\p\left(\sum_{j \neq i} W_j = D\right)}{\p\left(\sum_{j \neq i} W_j = D-1\right)} \leq \alpha \leq \max_i \frac{\p\left(\sum_{j \neq i} W_j = D\right)}{\p\left(\sum_{j \neq i} W_j = D-1\right)}\,. \] 
\end{theorem} 
\begin{proof}
The proof can be found in the Appendix, Section \ref{sec:prfThm1}. 
% For a fixed choice of $\boldsymbol x$, observe that $g(\boldsymbol x, \boldsymbol s)$ is monotonically decreasing in every component of $s$. Because 
% \[ g(\boldsymbol p, \boldsymbol \phi) = D \hspace{5mm} \text{ and } \hspace{5mm} g(\boldsymbol p, \boldsymbol \alpha) = D, \] 
% it follows immediately that $\alpha$ must lie between the largest and smallest value of $\phi_i$. 
\end{proof}

%\subsection{Main Results}

%Under Theorem \ref{thm:boundsOnAlpha}, if the $\phi_i$ are nearly constant across all values of $i$, then the logit shift will be a very good approximation to the conditional probability approach. We establish bounds on the values of the $\phi_i$ via the following theorem. 

%This emerges from a few key features of the Poisson Binomial distribution. As noted, $f(\alpha)$ is monotonic 

\begin{theorem}\label{thm:boundsOnPhi}
For any choice of $i \in \mathcal{V}$, we have
\[ \frac{\p\left(\sum_{j \in \mathcal{V}} W_j = D + 1\right)}{\p\left(\sum_{j \in \mathcal{V}} W_j = D\right)}\leq \frac{\p\left(\sum_{j \neq i} W_j = D\right)}{\p\left(\sum_{j \neq i} W_j = D-1\right)}\leq \frac{\p\left(\sum_{j \in \mathcal{V}} W_j = D\right)}{\p\left(\sum_{j \in \mathcal{V}} W_j = D-1\right)} \] 
\end{theorem}
\begin{proof}
The proof can be found in the Appendix, Section \ref{sec:prfThm2}.
\end{proof}
\subsection{Main Results}
The bounds from Theorem \ref{thm:boundsOnPhi} apply regardless of the choice of $i$, so we can combine the two theorems to find that
\begin{equation}\label{eq:double-bounds}
\begin{aligned}
 \frac{\p\left(\sum_{j \in \mathcal{V}} W_j = D + 1\right)}{\p\left(\sum_{j \in \mathcal{V}} W_j = D\right)} &\leq \min_i \frac{\p\left(\sum_{j \neq i} W_j = D\right)}{\p\left(\sum_{j \neq i} W_j = D-1\right)} \leq \alpha \\ & \leq  \max_i \frac{\p\left(\sum_{j \neq i} W_j = D\right)}{\p\left(\sum_{j \neq i} W_j = D-1\right)} \leq \frac{\p\left(\sum_{j \in \mathcal{V}} W_j = D\right)}{\p\left(\sum_{j \in \mathcal{V}} W_j = D-1\right)}.
\end{aligned}
\end{equation}

This is useful, because we can now use the outer bounds in Eq.~\ref{eq:double-bounds} to obtain a bound on the approximation error when estimating recalibrated scores $p_i^{\star}$ (obtained from the posterior update approach) via $\tilde p_i$ (obtained from the logit shift):

\begin{theorem}\label{thm:boundsOnError}
For large sample sizes, we obtain 
\[ \tilde p_i = p_i^{\star} + \mathcal{O} \left( \frac{1}{\sum_{j \in \mathcal{V}} p_j (1-p_j)} \right)\,. \] 
\end{theorem}
\begin{proof}
The proof can be found in the Appendix, Section \ref{sec:prfThm3}.

\end{proof}
Theorem~\ref{thm:boundsOnError} states that the error in using the logit shift approach to approximate the posterior recalibration update depends on the precision of the Poisson-Binomial distribution over sums of binary outcomes being aggregated (votes for the Democratic candidate, in our running example). As the variance of Poisson-Binomial deviates is maximal when all underlying probabilities are equal to 0.5, it follows that, in our running example, the approximation will perform best when voters in $\mathcal{V}$ are %similarly likely toss a fair coin 
equally likely to vote for either party.  As voters become more heterogenous, or as their support becomes more lopsided (or both, as would be the case in heavily polarized electorates), the quality of the approximation is expected to suffer. Fortunately, the binding bounds in Eq.~\ref{eq:double-bounds} are extremely tight for large enough samples, so that even in the worst case-scenarios, the approximation can be expected to perform well. We now briefly illustrate our analytical results with a small Monte-Carlo simulation.

\begin{comment}
\begin{theorem}
For large $N$, we have that the lower and upper bounds from Theorem \ref{thm:boundsOnPhi} differ by a multiplicative factor no larger than approximately $\left(\sum_j p_j (1- p_j) \right)^{-1}$.
\end{theorem}
\begin{proof}
It is a well-established result that, for large $N$, the Poisson-Binomial behaves approximately as a Normal random variable with the same mean and variance \citep[see e.g.][]{rosenman2018using}, namely 
\[ \mu = \sum_j p_j \hspace{5mm} \text{ and } \hspace{5mm} \sigma^2 = \sum_j p_j(1-p_j). \] 
Denote as $\psi(d)$ the density of a Normal distribution $\mathcal{N}(\mu, \sigma^2)$ with this mean and variance. We can simply compute 
\begin{align*}
\frac{\p\left(\sum_j W_j = D\right)}{\p\left(\sum_j W_j = D-1\right)}\bigg/ \frac{\p\left(\sum_j W_j = D + 1\right)}{\p\left(\sum_j W_j = D\right)} \approx \frac{\psi(D^2)}{\psi(D-1)\psi(D+1)} = \exp\left(\frac{1}{\sigma^2} \right) = \exp\left(\frac{1}{\sum_j p_j (1-p_j)} \right)\,.
\end{align*} 
For large $N$, $\sum_j p_j(1-p_j)$ will be large, and hence we can use a Taylor approximation: 
\[ \frac{\p\left(\sum_j W_j = D\right)}{\p\left(\sum_j W_j = D-1\right)}\bigg/ \frac{\p\left(\sum_j W_j = D + 1\right)}{\p\left(\sum_j W_j = D\right)} \approx 1 + \frac{1}{\sum_j p_j(1-p_j)}\,. \] 
\end{proof}
We have now established that the $\phi_i$ vary minimally across choices of $i$ for large $N$; it follows that $\alpha$ is a very good approximation for all of the $\phi_i$ and hence that our proposed procedure will approximate the true $p_i^{\star}$ very well. 
\end{comment}

\section{Simulations}

We conduct a brief simulation study to empirically demonstrate the efficacy of this approach. We simulate using $1,000$ units, a sample size at which computations of the true $p_i^{\star}$ are possible. We draw the initial probabilities $p_i$ according to the six distributions discussed in \cite{biscarri2018simple}. We then consider the cases in which the observed $D$ is either 20\% above or 20\% below the expectation, $\sum_i p_i$. 

We compute the true probabilities using Biscarri's algorithm as implemented in the \textit{PoissonBinomial} package \citep{junge2020package}, and compare it against the estimates obtained using our heuristic method. We report the RMSE as well as the proportion of variance in the true $p_i^{\star}$ that is \emph{not} explained by our method. Results are given in Table \ref{tab:simulations}. Across all settings, our approximations perform extremely well. 

\begin{table}[bt]
\label{tab:simulations}
\centering
\begin{tabular}{llrccc}
\toprule
\textbf{$\boldsymbol{p_i}$ Setting}  & \textbf{Sampling Distribution} & \textbf{Observed D} & \textbf{RMSE} & \textbf{1 $-$ R$\boldsymbol{^2}$ } \\ \midrule
Uniform               & Uniform(0, 1)                             & $-$20\%               & 0.00022       & 5.58$\times 10^{-7}$                          \\
Uniform               & Uniform(0, 1)                             & 20\%                & 0.00021       & 5.46$\times 10^{-7}$                          \\
Close to Zero         & Beta(0.1, 3)                        & $-$20\%               & 0.00047       & 3.41$\times 10^{-5}$                          \\
Close to Zero         & Beta(0.1, 3)                        & 20\%                & 0.00039       & 1.85$\times 10^{-5}$                          \\
Close to One          & Beta(3, 0.1)                        & $-$20\%               & 0.00036       & 1.22$\times 10^{-6}$                          \\
Close to One          & Beta(3, 0.1)                        & 20\%                & --            & --                                \\
Extremal              & 0.5*Beta(0.1, 3) + 0.5*Beta(3, 0.1) & $-$20\%               & 0.00048       & 1.14$\times 10^{-6}$                          \\
Extremal              & 0.5*Beta(0.1, 3) + 0.5*Beta(3, 0.1) & 20\%                & 0.00048       & 1.12$\times 10^{-6}$                          \\
Central               & Beta(3, 3)                          & $-$20\%               & 0.00015       & 6.86$\times 10^{-7}$                          \\
Central               & Beta(3, 3)                          & 20\%                & 0.00015       & 6.86$\times 10^{-7}$                          \\
Bimodal               & 0.5*Beta(3, 10) + 0.5*Beta(10, 3)   & $-$20\%               & 0.00023       & 6.92$\times 10^{-7}$                          \\
Bimodal               & 0.5*Beta(3, 10) + 0.5*Beta(10, 3)   & 20\%                & 0.00023       & 6.86$\times 10^{-7}$ \\          \bottomrule
\end{tabular}
\caption{Approximation error, as measured by RMSE and $1 - R^2$, using our heuristic method vs. the true Poisson-Binomial probabilities, under various settings. No results are reported in the row in which $1.2 \times \sum_i p_i$ would exceed the sample size of $1,000$.}
\end{table}

\section{Discussion}

In this paper, we have considered the problem of updating voter scores to match observed vote totals from an election. We have shown that the relatively simple ``logit shift" algorithm is a very good approximation to computing the true conditional probability. This is an especially useful insight for campaign analysts and researchers alike, because the logit shift is significantly more efficient computationally than the calculation of the exact posterior recalibration update. 

 It is worth being explicit about the limitations of this approach. Under the posterior update model, we treat the original scores $p_i$ as a prior over Democratic vote probability. In turn, the updated scores $p_i^{\star}$ deviate from the initial scores only by assuming the observed vote tallies deviate from the expectation of $\sum_i p_i$ due to random error. Crucially, the updated probabilities retain the same ordering as the prior probabilities, which implies the original scoring model must discriminate positive and negative (but unobservable, in the case of voting) individual cases well. It is also important to note that the realization of $\mathcal{V}$ over which values of $D$ are defined can have an impact on the quality of the approximation: the approximation will be better when the number of Democratic votes $D$ tallies the choices of voters in very competitive districts than when it tallies votes in landslide ones, and choosing a level of aggregation with too few voters in it could render the error bounds too loose. In most practical instances, however, the logit shift can be expected to perform very well.

Hence, this approach represents a useful -- albeit crude -- method of updating individual-level scores to incorporate information from a completed election. More complex insights about the electorate, such as the marked underperformance of Democrats among Hispanics voters in the 2020 election, cannot be directly incorporated by computing the posterior probabilities (or their approximation via the logit shift). Methods based on ecological inference \citep[e.g.][]{king2004ecological} would be necessary to capture this structure. Such methods represent a promising potential extension of the insights provided in this manuscript.

\bibliographystyle{apalike} % This yields author(year) citations. 
\bibliography{refs}      % bib for this and related papers 

\appendix

\section{Characterizations of the Logit Shift}\label{sec:char}

First, we show that the logit shift minimizes the summed KL divergence between $\tilde p_i$ and $p_i$, subject to $\sum \tilde p_i = D$ constraint. 

Define the minimum-KL-divergence optimization problem as 
\begin{equation}\label{optProb}
    \begin{aligned}
     \text{minimize} & \hspace{5mm} \sum_{i \in \mathcal{V}}- \tilde x_i \log\left(\frac{p_i}{\tilde x_i} \right) - (1 - \tilde x_i) \log \left( \frac{1-p_i}{1- \tilde x_i} \right)  \\
     \text{subject to} & \hspace{5mm} \sum_{i \in \mathcal{V}} \tilde x_i = D, \hspace{3mm} 0 \leq \tilde x_i \leq 1 \text{ for } i \in \mathcal{V}.  
    \end{aligned}
\end{equation}
The Lagrangian for Optimization Problem \ref{optProb} is 
\begin{align*}
L\left( \{\tilde x_i\}, \{\lambda_i\}, \{\nu_i\}, \gamma\right) &=  \sum_{i \in \mathcal{V}} - \tilde x_i \log\left(\frac{p_i}{\tilde x_i} \right) - (1 - \tilde x_i) \log \left( \frac{1-p_i}{1- \tilde x_i} \right) +  \\ & \sum_{i \in \mathcal{V}} \lambda_i (\tilde x_i - 1) - \nu_i \tilde x_i + \gamma \left( \sum_{i \in \mathcal{V}} \tilde x_i - D \right) \,.
\end{align*} 
We make the standard assumptions that $0 < p_i < 1$ for all $i \in \mathcal{V}$ and $0 < D < |\mathcal{V}|$. Define the point $\left( \{\tilde x_i\}, \{\lambda_i\}, \{\nu_i\}, \gamma\right) = \left( \{\tilde p_i\}, \{0\}, \{0\}, \log(\alpha) \right)$. We consider the Karush-Kuhn-Tucker (KKT) conditions at this point. For the Lagrangian gradient condition, observe:  
\begin{align*}
\nabla L\left( \{\tilde p_i\}, \{0\}, \{0\}, \log(\alpha)\right) & = -\log \left(\frac{p_i/(1-p_i) }{\tilde p_i/(1-\tilde p_i)}\right) + \log(\alpha) \\
&= -\log \left(\frac{p_i/(1-p_i) }{p_i/(\alpha(1-p_i))}\right) + \log(\alpha) \\
&= 0\,,
\end{align*}
while the other four KKT conditions are automatically satisfied at this point. It follows that this point is dual optimal. Lastly, because the objective function is convex and there exist choices of $x_i$ satisfying $0 < x_i < 1$ for $i \in \mathcal{V}$ and $\sum_{i \in \mathcal{V}} x_i = D$, strong duality is attained. Hence, our point is optimal and the $\tilde p_i$ are a solution to Optimization Problem \ref{optProb}. For background technical details, see \cite{boyd2004convex}. 

Next, we show that this procedure is equivalent to an intercept shift if the $p_i$ are defined based on a logistic regression, i.e. 
\[ p_i = \frac{\exp(\beta_0 + X_i\beta)}{1 + \exp(\beta_0 + X_i\beta)\,.} \]
where $\beta_0 \in \mathbb{R}$ is the intercept, $X_i \in \mathbb{R}^p$ is the covariate vector for unit $i$, and $\beta \in \mathbb{R}^p$ is the coefficient vector. Plugging this expression into the definition of $\tilde p_i$, we obtain \begin{align*}
    \tilde p_i &= \frac{1}{1 + \exp(-\beta_0 - X_i \beta)\alpha} \\
    &= \frac{1}{1 + \exp(-(\beta_0 - \log(\alpha)) - X_i \beta)} \\
    &= \frac{\exp(\beta_0 - \log(\alpha) + X_i\beta)}{1 + \exp(\beta_0 - \log(\alpha) + X_i\beta)} \,,
\end{align*}
where on the last line we see that $\tilde p_i$ is formulated as a logistic regression with an identical coefficient vector and its intercept shifted by $-\log(\alpha)$. 

\section{Proof of Theorem \ref{thm:boundsOnAlpha}}\label{sec:prfThm1}
For a fixed choice of $\boldsymbol x$, observe that $g(\boldsymbol x, \boldsymbol s)$ is monotonically decreasing in every component of $s$. Denote $\boldsymbol \alpha = \{\alpha\}_{i \in \mathcal{V}}$, the vector repeating $\alpha$ a total of $|\mathcal{V}|$ times. Because
\[ g(\boldsymbol p, \boldsymbol \phi) = D \hspace{5mm} \text{ and } \hspace{5mm} g(\boldsymbol p, \boldsymbol \alpha) = D, \] 
it follows immediately that $\alpha$ must lie between the largest and smallest value of $\phi_i$ across all choices of $i$. 

\section{Proof of Theorem \ref{thm:boundsOnPhi}}\label{sec:prfThm2}
The log-concavity of the Poisson Binomial distribution is a well-established result \citep[see e.g.][]{wang1993number}. Hence, for any choice of $i$, we have 
\begin{equation}\label{eq:logConcave}
     \p\left(\sum_{j \neq i}W_j = D - 2 \right) \p\left(\sum_{j \neq i}W_j = D  \right) \leq \p\left(\sum_{j \neq i}W_j = D - 1 \right)^2 
\end{equation}
Multiplying both sides of the equality by $p_i$, and adding the same quantity to both sides, we obtain an updated inequality 
\begin{align*}
&p_i \p\left(\sum_{j \neq i}W_j = D - 2 \right) \p\left(\sum_{j \neq i}W_j = D  \right) + (1- p_i) \p\left(\sum_{j \neq i}W_j = D - 1 \right) \p\left(\sum_{j \neq i}W_j = D  \right) \leq \\
&p_i \p\left(\sum_{j \neq i}W_j = D - 1 \right)^2 + (1- p_i) \p\left(\sum_{j \neq i}W_j = D - 1 \right) \p\left(\sum_{j \neq i}W_j = D  \right).
\end{align*} 
Collecting terms, we get 
\begin{equation}\label{bigIneq}
\begin{aligned}
& \p\left(\sum_{j \neq i}W_j = D  \right) \left( p_i \p\left(\sum_{j \neq i}W_j = D - 2 \right) + (1- p_i) \p\left(\sum_{j \neq i}W_j = D -1  \right)\right) \leq \\ & \p\left(\sum_{j \neq i}W_j = D - 1 \right) \left(p_i \p\left(\sum_{j \neq i}W_j = D - 1 \right) +  (1- p_i) p\left(\sum_{j \neq i}W_j = D  \right)\right) 
\end{aligned}
\end{equation}
The terms in parentheses can be collapsed into a single Poisson Binomial probability, making use of the recursion defined in Footnote 1%~\ref{fn:recursion}
. Subbing these expressions into Inequality \ref{bigIneq}, we obtain
\[ \p\left(\sum_{j \neq i}W_j = D  \right)  \p\left(\sum_{j \in \mathcal{V}}W_j = D -1  \right)  \leq \p\left(\sum_{j \neq i}W_j = D - 1 \right)  \p\left(\sum_{j \in \mathcal{V}}W_j = D  \right) \,, \] 
which yields the upper bound in Theorem 2.

The proof of the lower bound proceeds by incrementing $D$ by 1 in Inequality \ref{eq:logConcave} and following the same set of steps.

\section{Proof of Theorem \ref{thm:boundsOnError}}\label{sec:prfThm3}
Under the bounds defined in Inequality 4 in the main text,
%~\ref{eq:double-bounds} %(\ref{bigIneq})
we obtain the following approximation bounds for $\alpha_i$: 
\[ \frac{|\alpha - \phi_i|}{\phi_i} \leq \frac{ \frac{\p\left(\sum_{j \in \mathcal{V}} W_j = D\right)}{\p\left(\sum_{j \in \mathcal{V}} W_j = D-1\right)} - \frac{\p\left(\sum_{j \in \mathcal{V}} W_j = D + 1\right)}{\p\left(\sum_{j \in \mathcal{V}} W_j = D\right)}}{\frac{\p\left(\sum_{j \in \mathcal{V}} W_j = D + 1\right)}{\p\left(\sum_{j \in \mathcal{V}} W_j = D\right)}} \,.\]
It is a well-established result that, for large $N$, the Poisson-Binomial behaves approximately as a Normal random variable with the same mean and variance \citep[see e.g.][]{siripraparat2021local}, namely 
\[ \mu = \sum_j p_j \hspace{5mm} \text{ and } \hspace{5mm} \sigma^2 = \sum_j p_j(1-p_j). \] 
Denote as $\psi(d)$ the density of a Normal distribution $\mathcal{N}(\mu, \sigma^2)$ with this mean and variance, evaluated at $d$. \cite{siripraparat2021local} show that over all possible choices of $d$, the largest deviation between $\psi(d)$ and $\p\left(\sum_{j \in \mathcal{V}} W_j = d\right)$ is bounded above by $C_1/\sigma^2$ for a constant $C_1 > 0$. Hence 
\begin{equation}\label{boundOnAlpha}
\begin{aligned}
 \frac{|\alpha - \phi_i|}{\phi_i} &\leq \frac{\psi(D)^2}{\psi(D-1)\psi(D+1)} - 1 + \mathcal{O}\left(\frac{1}{\sigma^2} \right) \\
&= \exp\left(\frac{1}{\sigma^2} \right) - 1 + \mathcal{O}\left(\frac{1}{\sigma^2} \right) = \mathcal{O} \left( \frac{1}{\sigma^2} \right) 
\end{aligned}
\end{equation}

Lastly, we observe 
\begin{align*}
\tilde p_i &= \frac{1}{1 + \frac{1-p_i}{p_i} \alpha} \\
 &= \frac{1}{1 + \frac{1-p_i}{p_i} \phi_i \left( 1 + \frac{\alpha - \phi_i}{\phi_i}\right)} \\
&= \frac{1}{1 + \frac{1-p_i}{p_i} \phi_i} - \frac{(1-p_i)p_i \phi_i}{(p_i +
\phi_i - p_i \phi_i)^2} \left( \frac{\alpha - \phi_i}{\phi_i}\right) + \mathcal{O} \left(\left(\frac{\alpha - \phi_i}{\phi_i} \right)^2 \right) \\
&= p_i^{\star} + \mathcal{O}\left(\frac{1}{\sigma^2}\right)\,,
\end{align*}
where the last line follows by plugging in the bound on $(\alpha - \phi_i)/\phi_i$ from Inequality \ref{boundOnAlpha} and observing
\[ \left| \frac{(1-p_i)p_i \phi_i}{(p_i + \phi_i - p_i \phi_i)^2}\right| \leq \frac{1}{4} \hspace{5mm} \text{for $0 < p_i < 1, 0 < \phi_i$.} \]

\end{document}